\documentclass[12pt,a4paper]{article}

\usepackage[utf8]{inputenc}
\usepackage{enumerate}
\usepackage{amsmath}
\usepackage{amsfonts}
\usepackage{amssymb}
\usepackage{amsthm}
\usepackage{mathtools}
\usepackage{tikz}
\usetikzlibrary{shapes.geometric,calc}
\usepackage{hyperref}
\usepackage{authblk}

\usepackage{algorithm}
\usepackage{algpseudocode}
\algnewcommand\algorithmicinput{\textbf{Input:}}
\algnewcommand\algorithmicoutput{\textbf{Output:}}
\algnewcommand\Input{\item[\algorithmicinput]}
\algnewcommand\Output{\item[\algorithmicoutput]}

\usepackage[capitalize]{cleveref}

\usepackage[backgroundcolor=gray!10]{todonotes}

\usepackage[backend=biber,isbn=false,url=false,doi=true,maxcitenames=2,maxbibnames=99]{biblatex}
\bibliography{bibliography}

\usepackage{macros}

\usepackage{fullpage}
\newtheorem{theorem}{Theorem}[section]

\newtheorem{observation}[theorem]{Observation}
\newtheorem{lemma}[theorem]{Lemma}
\newtheorem{corollary}[theorem]{Corollary}
\newtheorem{definition}[theorem]{Definition}

\newtheorem{proposition}[theorem]{Proposition}

\newcommand{\pvg}{terrain visibility graph}
\newcommand{\Pvg}{Terrain visibility graph}

\newcommand{\pg}{persistent graph}

\newcommand{\xprop}{X-property}
\newcommand{\barprop}{bar-property}
\DeclareMathOperator{\skel}{skel}
\newcommand{\id}{\mathrm{id}}

\newcommand{\triangs}[1]{\mathcal{T}_{#1}}
\newcommand{\persis}[1]{\mathcal{P}_{#1}}

\tikzset{
	vertex/.style={circle,draw, inner sep = 1.5pt},
	edge/.style={, color=black},
	curveedge/.style={color=black,out=90,in=90},
	hyperplane/.style={line width=2pt},
}

\title{Persistent Graphs and Cyclic Polytope Triangulations}

\author[1]{Vincent Froese}
\author[1]{Malte Renken\footnote{Supported by the DFG project NI~369/17-1.}}

\affil[1]{\small
  Technische Universit\"at Berlin, Faculty~IV, Institute of Software Engineering and Theoretical Computer Science, Algorithmics and Computational Complexity, Berlin, Germany,\protect\\
  \{vincent.froese, m.renken\}@tu-berlin.de}

\begin{document}

\maketitle

\begin{abstract}
  We prove a bijection between the triangulations of the 3-dimensional cyclic polytope~$C(n+2,3)$ and persistent graphs with~$n$ vertices.
  We show that under this bijection the Stasheff-Tamari orders on triangulations naturally translate to subgraph inclusion between persistent graphs. Moreover, we describe a connection to the second higher Bruhat order~$B(n,2)$.
  We additionally give an algorithm to efficiently enumerate all persistent graphs on~$n$ vertices and thus all triangulations of~$C(n+2,3)$.
  \medskip
  
\end{abstract}

\section{Introduction}

In this work we prove a one-to-one correspondence between the triangulations of 3\nobreakdash-dimen\-sional cyclic polytopes and persistent graphs.
Cyclic polytopes are natural generalizations of convex polygons to higher dimensions
and are among the most studied classes of polytopes.
They are neighborly and achieve the maximum number of faces according to the upper bound theorem of \textcite{mcmullen_1970}.
Triangulations of the $d$-dimensional cyclic polytope~$C(n,d)$ are well-studied~\cite{RV12}.
It is known that a triangulation is fully determined by the set of its $\lfloor d/2 \rfloor$-dimensional faces~\cite{Dey93}.
For even dimension~$d$, a combinatorial description of this set is known~\cite{OT12}.
For odd dimension, however, no characterization is known so far.

We give a characterization for~$d=3$ by proving a bijection between the set of triangulations of~$C(n,3)$ and the class of so-called \emph{\pg{}s}.
These (vertex-ordered) graphs are defined by three simple combinatorial properties:
the existence of a Hamilton path,
the rule that two crossing edges force the existence of a third edge,
and the fact that any two adjacent, non-consecutive vertices have a common neighbor between them
(see \cref{def:persistent_graphs} for a precise definition).
Persistent graphs are probably best known for their long-conjectured equality to the set of \emph{\pvg{}s},
which has been recently disproved \autocite{ameer2020terrain}.
A graph $G$ is a \emph{\pvg{}} (sometimes referred to as 1.5-dimensional \pvg{}), if there exists a sequence of points $p_i \in \RR^2$ with ascending $x$-coordinates (the vertices of~$G$)
such that there is an edge between $p_i$ and $p_j$ if and only if the line segment connecting $p_i$ and $p_j$ does not pass below any other point in between.
\Pvg{}s are known to be persistent and for all sufficiently small persistent graphs the reverse implication also holds.
But recently \textcite{ameer2020terrain} constructed a \pg{} with 35 vertices and proved that it is not a \pvg{}.
However, \textcite{AbelloStaircaseI}, and later \textcite{Saeedi2015}, showed that \pg{}s are exactly the visibility graphs of \emph{pseudo-terrains},
in which straight lines are replaced by \emph{pseudo-lines} which need not be straight but must still intersect pairwise exactly once (see also \autocite{RourkePseudoVisibility}).

Our result now connects \pg{}s to a different class of geometrical objects and provides new insights into the combinatorial structure of cyclic polytope triangulations.
This implies more efficient enumeration (and also random generation) of triangulations of the 3-dimensional cyclic polytope.

\subsection{Related Work}
For a general overview on triangulations, see the monograph by \textcite{DRS10}.
The triangulations of cyclic polytopes and their poset structures have received considerable interest \cite{KV91,ER96,Rambau97,RV12,OT12}.
Also, efforts have been made to determine the number of triangulations~\cite{AS02,RV12,JK18}.
\textcite{Thomas02} gave a bijection between the triangulations of the cyclic polytope~$C(n, d)$ and so-called \emph{snug partitions} of the set~$[n-1]^d$.

\Pvg{}s are closely related to visibility graphs of so-called staircase polygons~\cite{Colley92}
as terrains can be turned into staircase polygons by adding reflex vertices between consecutive terrain vertices.
In this context, they were studied by~\textcite{AbelloStaircaseI}, who proved that they are persistent
and gave an algorithm to construct a so-called \emph{balanced tableau} from a persistent graph, from which a pseudo-terrain can be obtained.
Even though the authors claimed that conversely every persistent graph is a \pvg{}, this was recently disproved~\autocite{ameer2020terrain}.
A simplified proof of the results of~\textcite{AbelloStaircaseI} was published by \textcite{Saeedi2015},
who also showed that persistent graphs correspond to a restricted class of 3-signotopes.
Some graph-theoretic results regarding (forbidden) induced subgraphs of \pvg{}s and relation to other graph classes are known~\cite{FR19}.
Interestingly, in the context of time series data, \pvg{}s (there called time series visibility graphs) have received a lot of attention as an analytical tool~\cite{Lacasa4972} (see also references in~\cite{FR19}).
Also related classes such as \pvg{}s with uniform step length \cite{AbelloStaircaseUniform} and \emph{horizontal visibility graphs} \cite{GUTIN20112421} (a subclass of \emph{bar visibility graphs}) have been individually studied (the latter are shown to be exactly the outerplanar graphs containing a Hamilton path).

\subsection{Preliminaries}

We introduce some notation, basic definitions and preliminary results.

\paragraph{Notation.}
We define~$[n] := \{1,\ldots,n\}$ and denote the set of all size-2 subsets of~$[n]$ by~$\binom{[n]}{2}$.
The convex hull of a set~$S$ of points is denoted~$\conv(S)$.
We assume the reader to be familiar with the basics of the theory of polytopes (see e.g.~\textcite{Ziegler95}).
For a polytopal complex~$C$, we denote the set of $i$-dimensional faces of~$C$ as $F_i(C)$ and we write $f_i(C) := \abs{F_i(C)}$. The $i$-skeleton of~$C$ is defined as $\skel_i(C)=\bigcup_{j=0}^iF_j(C)$. Note that the 1-skeleton defines a graph with vertices~$F_0(C)$
and edges~$F_1(C)$.
Throughout this work, we always consider combinatorial faces and simplices, that is, we only consider the corresponding vertex sets.

\paragraph{Cyclic Polytopes.}

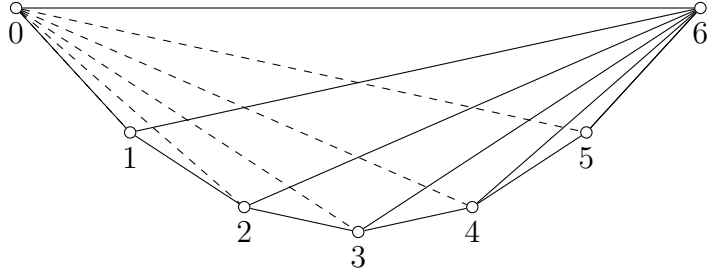
\begin{figure}[t]
  \centering
  \begin{tikzpicture}[xscale=1.5,yscale=0.33]
    \tikzset{every node/.style={vertex}}
    \node[label=below:0] (1) at (-3,9) {};
    \node[label=below:1] (2) at (-2,4) {};
    \node[label=below:2] (3) at (-1,1) {};
    \node[label=below:3] (4) at (0,0) {};
    \node[label=below:4] (5) at (1,1) {};
    \node[label=below:5] (6) at (2,4) {};
    \node[label=below:6] (7) at (3,9) {};
    \draw (1) -- (2) -- (3) -- (4) -- (5) -- (6) -- (7);
    \foreach \x in {1,2,...,6} {
      \draw (\x) -- (7);
    }
    \foreach \x in {2,3,...,6} {
      \draw[dashed] (1) -- (\x);
    }
  \end{tikzpicture}
  \caption{Schematic drawing of $C(7,3)$ on the~$x$-$y$-plane (dashed lines lying below). }
  \label{fig:example_C73}
\end{figure}

For an integer~$d \ge 1$, the $d$-dimensional cyclic polytope is defined via
 the $d$-th \emph{moment curve}:
 \[\mu_d \colon \mathbb{R} \to \mathbb{R}^d,\; t \mapsto (t,t^2,\ldots,t^d).\]
 Let $t_1 < t_2 < \ldots < t_n$ be $n > d$ real numbers. Then,
 \[C(n,d) \coloneqq \conv\{\mu_d(t_1),\ldots,\mu_d(t_n)\}\]
 is the $d$-dimensional cyclic polytope with~$n$ vertices.
 It is well-known that the combinatorics of~$C(n,d)$ do not depend on the particular values of~$t_1,\ldots,t_n$ but just on the number~$n$.
 In the remainder of this work, we consider~$C := C(n+2,3)$ and denote its vertices by $0, 1, \dots, n+1$, ordered by their first coordinate.
The faces of~$C$ are determined by \emph{Gale's evenness criterion}~\cite[Theorem~3]{Gale63} as follows (see \Cref{fig:example_C73} for an example):
\begin{align*}
  F_1(C) &= \set{\{0,n+1\}, \{0,i\},\{i,n+1\},\{i,i+1\} ; 0< i < n+1},\\
  F_2(C) &= \set{\{0,i,i+1\}, \{i-1,i,n+1\} ; 0 < i < n+1}.
\end{align*}

A \emph{triangulation} of~$C$ is a collection~$T=\{S_1,\ldots,S_m\}$ of 3-simplices (that is, tetrahedra)~$S_i=\{a,b,c,d\}\subseteq\{0,\ldots,n-1\}$, such that
$\bigcup_{i=1}^m\conv(S_i) = C$ and each pair of 3-simplices intersects in a common (possibly empty) face.
We denote the set of all triangulations of~$C = C(n+2, 3)$ by~$\triangs{n+2}$.

We proceed with some known results about characterizing triangulations of~$C$.
We will use these in order to prove our main result.
A \emph{circuit} (also called a \emph{primitive Radon partition}) is a pair~$(X,Y)$ of disjoint minimal subsets of vertices of~$C$ such that
$\conv(X) \cap \conv(Y) \neq \emptyset$.
The circuits of~$C$ are easily characterized as follows:
\begin{lemma}[\cite{Breen1973}]\label{thm:cyclic_circuits}
	The circuits of $C$ are exactly the pairs
	$(\{x_1, x_3, x_5\}, \{x_2, x_4\})$ with $x_1 < x_2 < x_3 < x_4 < x_5$.
\end{lemma}

The above result on circuits allows us to give the following characterization of a triangulation of~$C$ as a direct consequence of \cite[Proposition~2.2]{Rambau97}.

\begin{proposition}\label{thm:triangulation_conditions}
	A set $T$ of 3-simplices with vertices from $C$
	is a triangulation of~$C$ if and only if
	\begin{enumerate}
		\item for each $S \in T$ and each facet $F$ of $S$ either $F \in F_2(C)$ or there is another 3-simplex $S' \in T$ of which $F$ is a facet (Union-Property), and
		\item there is no pair of 3-simplices $S, S' \in T$ such that
                  $\{x_1, x_3, x_5\} \subset S$ and $\{x_2, x_4\} \subset S'$ for any $x_1 < x_2 < x_3 < x_4 < x_5$ (Intersection-Property).
	\end{enumerate}
\end{proposition}

The next observation states that an internal edge of a triangulation is contained in at least three 3-simplices.
It follows directly from the union-property of \Cref{thm:triangulation_conditions}.

\begin{observation}\label{thm:internal_edge_surrounded}
	Let $T$ be a triangulation of $C$ and let $\{v, w\}$ be an internal edge, i.e., $\{v, w\} \in F_1(T) \setminus F_1(C)$.
	Then, there are $k \geq 3$ vertices $x_1, \ldots, x_k, x_{k+1} = x_1$ such that
	$\{v, w, x_i, x_{i+1}\} \in T$ for all $i = 1, \dots, k$.
\end{observation}

This leads us to the following helpful lemma about internal edges.

\begin{lemma}\label{thm:super_bar_property}
	Let $T$ be a triangulation of~$C$ and let $\{v, w\} \in F_1(T) \setminus F_1(C)$ be an internal edge with $v < w$.
	Then, there are vertices $a, b, c$ with $a < v < b < w < c$ such that
	$\{\{v, w, a\},\{v, w, b\},\{v, w, c\}\}\subseteq F_2(T)$.
\end{lemma}
\begin{proof}
	Note that~$0 < v$ and~$w < n+1$ since~$\{v,w\}$ is not an edge of~$C$.
	Let $x_1, \ldots, x_{k+1}=x_1$ be the $k\ge 3$ vertices given by \cref{thm:internal_edge_surrounded}.
	Then $K := \conv\{v, w, x_1, \dots, x_k\}$ is a cyclic polytope on $k+2$ vertices
	and $T' := \set{\set{v, w, x_i, x_{i+1}}; i=1,\dots,k}$ is a triangulation of $K$.
	Note that $\{v, w\}$ is an internal edge of $T'$ and thus not an edge of $K$.
	This implies that $v$~and~$w$ can not be the first or the last vertex or two consecutive vertices of~$K$.
	Therefore, $\{x_1, \dots, x_k\}$ contains vertices $a, b, c$ with $a < v < b < w < c$ as claimed.
\end{proof}

\paragraph{Terrain Visibility and Persistent Graphs.}
\label{def:persistent_graphs}

\Pvg{}s are visibility graphs of 1.5-dimensional \emph{terrains}, that is, $x$-monotone polygonal chains in the plane defined by a set~$V\subseteq\mathbb{R}^2$ of \emph{terrain vertices} with pairwise different~$x$-coordinates.
Two vertices~$v_1=(x_1,y_1)$ and~$v_2=(x_2,y_2)$ are adjacent if and only if they \emph{see each other}, that is, there is no vertex between them that lies on or above the line segment connecting them.
Formally, there exists an edge~$\{v_1,v_2\}$, for $x_1< x_2$, if and only if all terrain vertices~$(x,y)$ with $x_1 < x < x_2$ satisfy
\begin{align*}
y < y_1 + (x - x_1) \frac{y_2 - y_1}{x_2 - x_1}.
\end{align*}
\Cref{fig:example_tvg} depicts an example.
We denote the vertices by~$1,\ldots,n$ in increasing order of their~$x$-coordinates.

\begin{figure}
  \centering
  \begin{tikzpicture}[xscale=1.33]
    \tikzset{every node/.style={vertex}}
    \node[label=below:1] (1) at (0,1.5) {};
    \node[label=below:2] (2) at (1,0) {};
    \node[label=right:3] (3) at (2,1) {};
    \node[label=below:4] (4) at (2.5,-1) {};
    \node[label=above:5] (5) at (3.5,4) {};
    \node[label=below:6] (6) at (4.25,2) {};
    \node[label=below:7] (7) at (5,3) {};
    \draw[very thick] (1) -- (2) -- (3) -- (4) -- (5) -- (6) -- (7);
    \draw (1) -- (3);
    \draw (1) -- (5);
    \draw (2) -- (5);
    \draw (3) -- (5);
    \draw (5) -- (7);
  \end{tikzpicture}\hspace{1cm}
  \begin{tikzpicture}[baseline=-2cm]
    \tikzset{every node/.style={vertex}}
    \node[label=below:1] (1) at (1,0) {};
    \node[label=below:2] (2) at (2,0) {};
    \node[label=below:3] (3) at (3,0) {};
    \node[label=below:4] (4) at (4,0) {};
    \node[label=below:5] (5) at (5,0) {};
    \node[label=below:6] (6) at (6,0) {};
    \node[label=below:7] (7) at (7,0) {};
    \draw (1) -- (2) -- (3) -- (4) -- (5) -- (6) -- (7);
    \draw[curveedge] (1) to (3);
    \draw[curveedge] (1) to (5);
    \draw[curveedge] (2) to (5);
    \draw[curveedge] (3) to (5);
    \draw[curveedge] (5) to (7);
  \end{tikzpicture}
  \caption{A terrain visibility graph drawn in two different ways (with a corresponding terrain on the left).}
  \label{fig:example_tvg}
\end{figure}
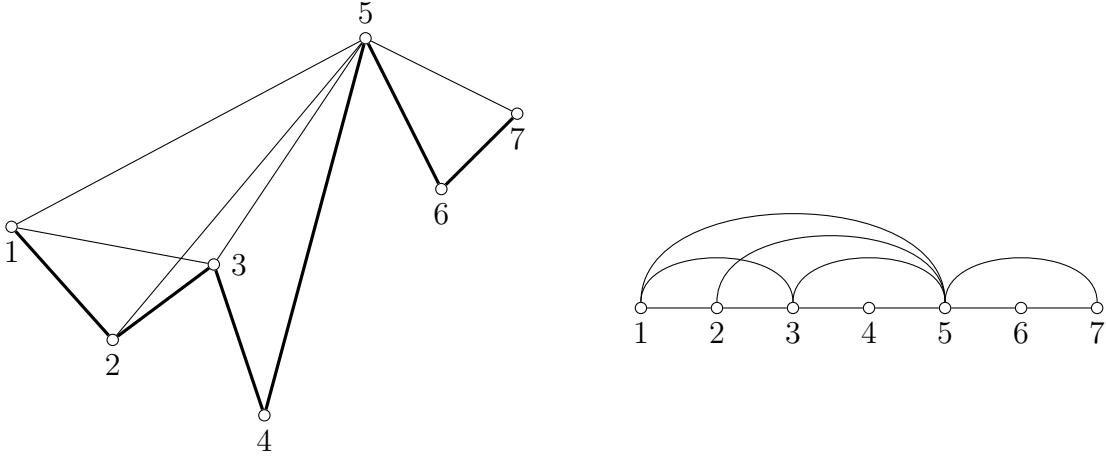

\pagebreak[2]

\Pvg{}s are known to be persistent~\cite{Saeedi2015}
where a graph $G=([n], E)$ is called \emph{persistent} if it satisfies the following three properties.
\begin{enumerate}
  \item It contains the Hamilton path~$1,\ldots,n$, that is,~$\{\{i,i+1\}\mid 1\le i < n\}\subseteq E$
  \item \textbf{\xprop{}}: If~$\{a,c\}\in E$ and~$\{b,d\}\in E$ for some vertices~$a < b < c < d$, then~$\{a,d\}\in E$.
  \item \textbf{\barprop{}}: For every edge~$\{a,b\}\in E$ with~$a < b-1$, there exists a vertex~$x$ with~$a < x < b$ such that~$\{a,x\}\in E$ and~$\{x,b\}\in E$.
\end{enumerate}
We denote the set of all persistent graphs with~$n$ vertices by~$\persis{n}$.
Note that graphs satisfying only the first two properties are called \emph{terrain-like}~\cite{FR20}.

We prove the following elementary property about \pg{}s,
which states that consecutive neighbors of a vertex are also neighbors of each other.
Here, $N(v)$ denotes the neighborhood of vertex~$v$ and~$N[v] = N(v)\cup\{v\}$.
\begin{lemma}\label{thm:consecutive_neighbors}
  Let $G=([n],E)$ be a \pg{} and let $a,b,c$ be vertices such that~$b < c$, $\{b,c\}\subseteq N(a)$, and there is no vertex $x \in N[a]$ with $b < x < c$. Then $\{b, c\} \in E$.
\end{lemma}
\begin{proof}
  We assume that $a < b < c$ (the case~$b < c < a$ is fully symmetric).
  By the \barprop{}, there exists a common neighbor~$x$ of~$a$ and~$c$ with~$a < x < c$.
  Note that~$b < x$ is not possible by assumption.
  If $x = b$, then we are done.
  Otherwise, we have~$a < x < b$ and the \barprop{} again implies the existence of a common neighbor~$x'$ of~$x$ and~$c$ with~$x < x' < c$.
  Now, if~$b < x'$, then the \xprop{} (applied to~$\{a,b\}$ and~$\{x,x'\}$) implies that~$x'$ is a neighbor of~$a$ which contradicts our assumption on~$b$ and~$c$.
  Thus, $a < x' \le b$. Note that we can repeat the above argument again on~$x'$ if~$x' < b$.
  Since $G$ is finite, we can conclude that $b$ is a neighbor of~$c$. 
\end{proof}

\section{A Bijection Between $\triangs{n+2}$ and $\persis{n}$}

In this section, we prove a bijection between triangulations of~$C(n+2,3)$ and persistent graphs on~$n$ vertices.
The central observation is that the 1-skeleton of a triangulation forms a persistent graph.
From this graph we will omit the outermost vertices $0$ and $n+1$ as they are always connected to all other vertices,
thus carrying no information.
See \Cref{fig:example_bijection} for an example.
Formally, the resulting map is defined as follows.

\begin{figure}
  \centering
  \begin{tikzpicture}[xscale=1.5,yscale=0.33]
    \tikzset{every node/.style={vertex}}
    \node[label=below:0] (1) at (-3,9) {};
    \node[label=below:1] (2) at (-2,4) {};
    \node[label=below:2] (3) at (-1,1) {};
    \node[label=below:3] (4) at (0,0) {};
    \node[label=below:4] (5) at (1,1) {};
    \node[label=below:5] (6) at (2,4) {};
    \node[label=below:6] (7) at (3,9) {};
    \draw (1) -- (2) -- (3) -- (4) -- (5) -- (6) -- (7);
    \foreach \x in {1,2,...,6} {
      \draw (\x) -- (7);
    }
    \foreach \x in {2,3,...,6} {
      \draw[dashed] (1) -- (\x);
    }
    \draw[line width=2pt] (7) -- (2);
    \draw[line width=2pt] (7) -- (4);
    \draw[line width=2pt] (7) -- (6);
    \draw[line width=2pt] (2) -- (6);
    \draw[line width=2pt] (2) -- (4);
    \draw[line width=2pt] (4) -- (6);
  \end{tikzpicture}
   \begin{tikzpicture}[baseline=-2cm]
    \tikzset{every node/.style={vertex}}
    \node[label=below:1] (1) at (1,0) {};
    \node[label=below:2] (2) at (2,0) {};
    \node[label=below:3] (3) at (3,0) {};
    \node[label=below:4] (4) at (4,0) {};
    \node[label=below:5] (5) at (5,0) {};
    \draw (1) -- (2) -- (3) -- (4) -- (5);
    \draw[curveedge,very thick] (1) to (3);
    \draw[curveedge,very thick] (3) to (5);
    \draw[curveedge,very thick] (1) to (5);
    
  \end{tikzpicture}
  \caption{Example of a triangulation~$T$ of~$C(7,3)$ and the corresponding persistent graph~$\Gamma(T)$ on five vertices. 
  The 3-simplices of the triangulation are $T=\{\{0,1,2,3\},\: \{1,2,3,6\},\: \{0,3,4,5\},\: \{3,4,5,6\},\: \{0,1,5,6\},\: \{0,1,3,5\},\:\{1,3,5,6\}\}$. 
  The \mbox{3-simplex}~$\{1,3,5,6\}$ (thick lines) yields the edges~$\{1,3\}$, $\{3,5\}$, and $\{1,5\}$ in~$\Gamma(T)$.
  Conversely, this 3-simplex is obtained from the edge~$\{3,5\}$ according to the inverse map~$\Xi$.}
  \label{fig:example_bijection}
\end{figure}
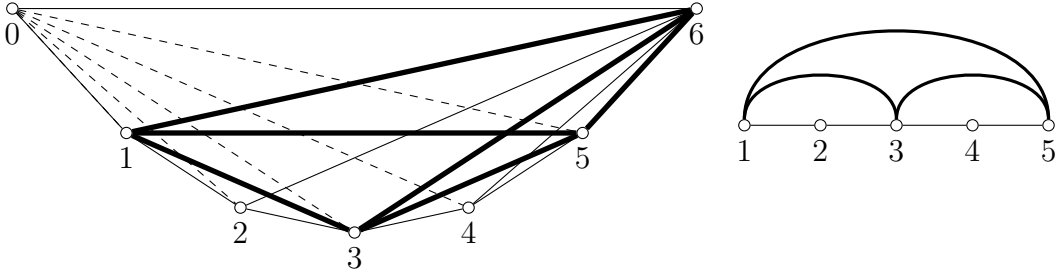

\begin{definition}
  For~$n\ge 2$, the map $\Gamma \colon \triangs{n+2} \rightarrow \persis{n}$ is defined as \[\Gamma(T) := \left([n],F_1(T)\cap [n]^2\right),\]
  that is, two vertices~$i$ and~$j$ are adjacent in~$\Gamma(T)$ if and only if~$\{i,j\}\subseteq S$ for some 3-simplex~$S\in T$.
\end{definition}

First, we show that~$\Gamma$ is well-defined, that is, $\Gamma(T)$ is in fact a persistent graph.

\begin{lemma}
	For every~$T\in\triangs{n+2}$, it holds $\Gamma(T)\in \persis{n}$.
\end{lemma}
\begin{proof}
  Clearly,~$\Gamma(T)$ contains the Hamilton path~$1,\dots,n$ since~$\{i,i+1\}\in F_1(C)$ and thus~$\{i,i+1\}\in F_1(T)$ for each~$i\in[n-1]$.
  
	Next, we show that $\Gamma(T)$ satisfies the \barprop.
	Let $e = \{v, w\}$ be an edge of~$\Gamma(T)$ with $v < w -1$.
	Then, $e$ is an internal edge, that is, $e \in F_1(T) \setminus F_1(C)$.
	Hence, by \cref{thm:super_bar_property}, there exists a vertex~$b$ with $v < b < w$ such that $\{v,w, b\}\in F_2(T)$.
        Therefore, $\{v,b\}$ and $\{b, w\}$ are edges of $\Gamma(T)$.
	
	Now, for the \xprop, assume towards a contradiction that~$\Gamma(T)$ contains the edges $\{u, w\}$ and $\{v, x\}$ with $u < v < w < x$, but $\{u, x\} \notin E(\Gamma(T))$.
	Let $(u, v, w, x)$ be lexicographically minimal with this property.
	Note that~$\{u,w\}$ and $\{v,x\}$ are both internal edges. Thus, \cref{thm:super_bar_property} applied to~$\{v,x\}$ implies that there exists a vertex~$a < v$ such that $\{a,v,x\}\in F_2(T)$ (and thus $\{\{a,v\},\{a,x\}\}\subseteq E(\Gamma(T))$).
	By minimality of $v$, it follows $a \leq u$.
	If $a < u$, then $\{a, v, x\}$ and $\{u, w\}$ are subsets of two different 3-simplices of $T$, contradicting the intersection-property of \cref{thm:triangulation_conditions} (since $(\{a,v,x\},\{u,w\})$ is a circuit).
	Thus, it follows $a = u$, that is,~$\{u,x\}\in E(\Gamma(T))$, which is a contradiction.
\end{proof}

In order to show that~$\Gamma$ is a bijection, we next define a map that maps a persistent graph to a triangulation. We then prove that this map is the inverse of~$\Gamma$.
To start with, we define the following auxiliary graph.

\begin{definition}\label{def:supergraph}
	For a \pg{} $G = ([n], E)$,
	we define the supergraph $\hat{G}:=(\{0, \ldots, n+1\}, E \cup (\{0,n+1\}\times\{0,\ldots,n+1\}))$, that is, $\hat{G}$ contains two additional vertices that are connected to all other vertices.
\end{definition}

It is easy to see that $\hat{G}$ is a \pg{} since adding a vertex that is adjacent to all others cannot violate the X- or \barprop{}.
Using \Cref{def:supergraph}, we now introduce the inverse map~$\Xi$.

\begin{definition}
	Let $G=([n], E)$ be a \pg{}. For $e = \{v, w\} \in E$, $v < w$, we define
	\begin{align*}
		\ell_G(e) &:= \max\set{i \in V(\hat{G}) ; i < v \text{ and } \{i, w\} \in E(\hat{G})} , \text{ and}\\
		r_G(e) &:= \min\set{i \in V(\hat{G}) ; w < i \text{ and } \{i, v\} \in E(\hat{G})} .
	\end{align*}
	Further, we define the 3-simplex $\xi_G(e):=\{\ell_G(e), v, w, r_G(e)\}$
	and the map~$\Xi\colon \persis{n} \rightarrow \triangs{n+2}$ as \[\Xi(G) := \set{\xi_G(e); e \in E}.\]
	We omit the index $G$ whenever it is clear from the context.
\end{definition}

Note that, by construction of~$\hat{G}$, the vertices~$\ell(e)$ and~$r(e)$ always exist.
Moreover, by \cref{thm:consecutive_neighbors}, the vertices in $\xi(e)$ form a clique in $\hat{G}$.
We now show that~$\Xi$ is well-defined, that is,~$\Xi(G)$ is indeed a triangulation.
To this end, we show that~$\Xi(G)$ satisfies the union-property and the intersection-property
according to \Cref{thm:triangulation_conditions}.
We start with the intersection-property.

\begin{lemma}\label{thm:triangle_disjointness}
	Let $G$ be a \pg{} and let $a < b < c$ be vertices of a 3-simplex $S \in \Xi(G)$.
	Then, $G$ does not contain any edge $\{x, y\}$ with $a < x < b < y < c$.
\end{lemma}
\begin{proof}
  Assume towards a contradiction that there exists such an edge~$\{x,y\}$.
  We assume the vertices to be chosen such that $(c - a) + (y - x)$ is minimal.
	Note that $\hat{G}$ contains the edge $\{x, y\}$ and a clique on~$\{a, b, c\}$.
	By the \xprop{}, $\hat{G}$ then also contains the edges~$\{a, y\}$ and $\{x, c\}$.

        Let $e \in E(G)$ be an edge with $\xi(e) = S$.
	Then, by definition of~$\xi$, it follows that $e \nsubseteq \{a, b, c\}$.
        To see this, note that~$e=\{a,c\}$ is not possible since~$b\not\in\{\ell(e),r(e)\}$.
        Also $e = \{a, b\}$ is not possible, since $c\neq r(e)$ (since $y < c$ is also neighbor of $a$). Analogously, $e=\{b,c\}$ is not possible, since~$a\neq \ell(e)$.
	
	Therefore, assume \wilog{} that $e = \{b, b'\}$ with $b < b' < c$ and $a = \ell(e)$ and $c = r(e)$.
	We cannot have $b' > y$, because then we could replace $c$ by $b'$ and decrease $(c-a) + (y-x)$.
	Also $b' = y$ is impossible since then $a = \ell(e)$ would contradict the fact that $x$ is a neighbor of $y$. Hence, $b' < y$.
	From the \barprop{} it follows that $x$ and $y$ have a common neighbor $z$ between them.
	The minimality of $(c-a) + (y-x)$ implies $b \leq z \leq b'$ since otherwise $x$ or $y$ could be replaced by $z$ (see \Cref{fig:intersect_prop}).
	Then, by the \xprop{}, there exists the edge $\{b, y\}$ (contradicting $r(e) = c$) or the edge $\{x, b'\}$ (contradicting $\ell(e) = a)$.
\end{proof}

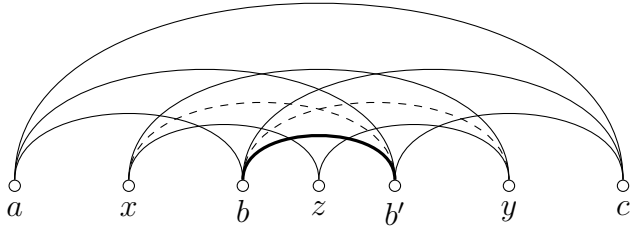
\begin{figure}
  \centering
  \begin{tikzpicture}
    \tikzset{every node/.style={vertex}}
    \node[label=below:$a$] (a) at (0,0) {};
    \node[label=below:$x$] (x) at (1.5,0) {};
    \node[label=below:$b$] (b) at (3,0) {};
    \node[label=below:$z$] (z) at (4,0) {};
    \node[label=below:$b'$] (b') at (5,0) {};
    \node[label=below:$y$] (y) at (6.5,0) {};
    \node[label=below:$c$] (c) at (8,0) {};
    \draw[curveedge] (a) to (c);
    \draw[curveedge] (b) to (c);
    \draw[curveedge,very thick] (b) to (b');
    \draw[curveedge] (x) to (y);
    \draw[curveedge] (a) to (b');
    \draw[curveedge] (x) to (z);
    \draw[curveedge] (y) to (z);
    \draw[curveedge,dashed] (x) to (b');
    \draw[curveedge,dashed] (b) to (y);
    \draw[curveedge] (a) to (b);
    \draw[curveedge] (b') to (c);
  \end{tikzpicture}
  \caption{The situation in \Cref{thm:triangle_disjointness} where $\ell(\{b,b'\})=a$ and $r(\{b,b'\})=c$. The existence of an edge~$\{x,y\}$ leads to a contradiction since it implies the existence of at least one of the dashed edges.}
  \label{fig:intersect_prop}
\end{figure}

Next, we prove that~$\Xi(G)$ satisfies the union-property.

\begin{lemma}\label{thm:image_triangulation}
  Let $G\in\persis{n}$ and let $S \in \Xi(G)$ be a 3-simplex containing vertices $a < b < c$ such that $\{a, b, c\}\not\in F_2(C)$.
  Then, there exists another 3-simplex~$S'\in\Xi(G)$ with~$\{a,b,c\}\subset S'$.
\end{lemma}
\begin{proof}
	Fix an edge $e \in E(G)$ with $S = \xi(e)$ (recall that the vertices $\xi(e)$ form a clique in~$\hat{G}$ by \Cref{thm:consecutive_neighbors}).
	By definition of~$\xi$, it holds $b \in e$.
	\Wilog{}, we can assume that either $e = \{a, b\}$ or $e = \{b, b'\}$ with $b < b' < c$ (the cases $e=\{b,c\}$ or $e = \{b, b'\}$ with $a < b' < b$ are symmetric).
	The following case distinction yields the existence of a vertex~$x$ with~$a < x < b$ such that~$x$ is a common neighbor of~$b$ and~$c$ in~$\hat{G}$.
	
	\textbf{Case 1:} $c = n+1$.
	Since $\{a, b, c\}$ is not a face of $C$, we have $a + 1 < b$.
	Clearly, the vertex $b-1$ is a neighbor of $b$ and~$c$ in~$\hat{G}$ (by construction).
	
	\textbf{Case 2:} $c < n+1$.
	If $\ell(\{b, c\}) = a$, then the 3-simplex $\xi(\{b, c\})$ also contains $\{a, b, c\}$ and we are done.
	Otherwise, $x=\ell(\{b,c\})$ is a common neighbor of $b$ and $c$ in~$\hat{G}$ by \Cref{thm:consecutive_neighbors}.

	In the following, we assume~$x$ to be chosen minimally.
	By \cref{thm:triangle_disjointness}, $x$ has no neighbor between $b$ and $c$. Thus, $r(\{x, b\}) = c$.
	Furthermore, $b$ has no neighbor between~$a$ and $x$ because, by the \xprop{}, this would also be a neighbor of $c$,
	contradicting the minimality of $x$.
	Therefore, $\ell(\{x, b\}) = a$ and thus, $\xi(\{x, b\})$ contains $\{a, b, c\}$ (note that $\{x,b\}\neq e$).
\end{proof}

\Cref{thm:triangle_disjointness,thm:image_triangulation} together with \Cref{thm:triangulation_conditions} now yield the following.

\begin{lemma}
  For every $G\in\persis{n}$, it holds $\Xi(G)\in\triangs{n+2}$.
\end{lemma}

Finally, we prove that~$\Xi$ is the inverse of~$\Gamma$.

\begin{theorem}\label{thm:bijection}
	The maps $\Gamma$ and $\Xi$ are mutually inverse bijections.
\end{theorem}
\begin{proof}
  First, we show that $\Gamma \circ \Xi = \id$.
  Let~$G=([n],E)$ be a persistent graph. Note that, by definition, for each edge~$e\in E$, the 3-simplex~$\xi(e)$ contains~$e$, that is,~$e\in F_1(\Xi(G))$. Thus, $E\subseteq E(\Gamma(\Xi(G)))$.
  Moreover, since the vertices in~$\xi(e)$ form a clique in~$\hat{G}$ (by \Cref{thm:consecutive_neighbors}), it follows that $(F_1(\xi(e)) \cap \binom{[n]}{2}) \subseteq E$. Thus, $E(\Gamma(\Xi(G))) \subseteq E$.
  Hence,~$\Gamma(\Xi(G)) = G$.
  
	To see that $\Xi \circ \Gamma = \id$, let $T$ be any triangulation of $C$ and let $S \in T$ be a 3-simplex.
	Let $a < b < c < d$ be the vertices of $S$.
	We claim that $a = \max\set{i; 0 \leq i < b, \{i,c\}\in F_1(T)}$.
	Assume towards a contradiction that there exists a vertex $x$ with $a < x < b$ and~$\{x,c\}\in F_1(T)$.
	Then, $\{a, b, d\}$ and $\{x, c\}$ are subsets of two different simplices of $T$, contradicting the intersection-property of \Cref{thm:triangulation_conditions}.
	By symmetry, we also obtain that $d = \min\set{i; c < i \leq n+1, \{i,b\}\in F_1(T)}$.
	Now, since~$\Gamma(T)$ contains the edge $\{b, c\}$, it follows that $\Xi(\Gamma(T))$ contains $S$. Thus, $T \subseteq \Xi(\Gamma(T))$.
	Since~$T$ and $\Xi(\Gamma(T))$ are triangulations of $C$ (by \cref{thm:image_triangulation}), this implies $T = \Xi(\Gamma(T))$.
\end{proof}

An interesting observation is that, for any $G \in \persis{n}$, the map $\xi_G \colon E(G) \rightarrow \Xi(G)$ is a bijection.
Its inverse is given by the map $\{a, b, c, d\} \mapsto \{b, c\}$, where $a < b < c < d$.
This implies that the number of edges in~$G$ equals the number of 3-simplices in~$\Xi(G)$.

To close this section, we compare our result for~$d=3$ with the characterization for even~$d$ by~\textcite{OT12}.
They showed that for every triangulation~$T$ of~$C(n+2,2k)$, the set of
$k$-dimensional faces of~$T$ that do not contain~$\{i,i+1\}$ for some~$i$ contains exactly~$\binom{n-k+1}{k}$ \emph{non-intertwining} tuples from~$\{0,\ldots,n+1\}^{k+1}$,
where~$(a_0,\ldots,a_k)$ \emph{intertwines} $(b_0,\ldots,b_k)$ if $a_0 < b_0 < a_1 < b_1 < \dots < a_k < b_k$.
Conversely, they also proved that every non-intertwining set of size~$\binom{n-k+1}{k}$ (which is maximal) defines a unique triangulation.
In the easiest case~$d = 2k = 2$, the triangulations of~$C(n+2,2)$ are simply all maximal outerplanar graphs (which are chordal).
Now, when moving to~$d=3$ dimensions, we lose planarity since edges can intertwine but any interwined pair has to satisfy the \xprop{}. Also, chordality is lost and replaced by the property that any cycle following the vertex order has to have a chord (which can serve as an equivalent replacement for the \barprop{} in the definition of persistent graph).

\section{Stasheff-Tamari Order on Persistent Graphs}

Classic objects studied in the context of triangulations of cyclic polytopes are the first and second \emph{Stasheff-Tamari orders}, which are certain partial orders on the set of triangulations.
In this section we show how these partial orders translate to partial orders on persistent graphs.
It is known that the first and second Stasheff-Tamari order are identical on~$\triangs{n+2}$~\cite{ER96}. Hence, we will only define and use the first Stasheff-Tamari order here.

Let $C = C(n+2, 3)$ and $W := \{v_1 < v_2 < \dots < v_5\}$ be a set of five vertices of $C$.
Note that $\conv(W)$ equals~$C(5,3)$ and has exactly two triangulations:
\begin{align*}
  T^* &:= \set{\{v_1, v_2, v_3, v_5\}, \{v_1, v_3, v_4, v_5\}} \text{ and}\\
  T_* &:= \set{\{v_1, v_2, v_4, v_5\}, \{v_1, v_2, v_3, v_4\}, \{v_2, v_3, v_4, v_5\}}.
\end{align*}

Now, let $T$ be a triangulation of $C$ with $T_* \subseteq T$.
Then, we obtain a new triangulation~$T'$ of $C$
via $T' := (T \setminus T_*) \cup T^*$.
In this case, we say that $T'$ is obtained from $T$ by a \emph{bistellar up-flip},
and conversely, $T$ is obtained from $T'$ by a \emph{bistellar down-flip}.
For any two triangulations $T, T'$ of $C$, we write $T \leq_1 T'$ if $T'$ is obtained from $T$ by a sequence of bistellar up-flips.
This defines a partial order called the first Stasheff-Tamari order~\cite{KV91}.
Note that $T \leq_1 T'$ implies that $\abs{T} \geq \abs{T'}$.

The following theorem shows that a bistellar up-flip corresponds to removing a certain edge from the corresponding persistent graph.

\begin{theorem}\label{thm:stasheff-tamari}
	Let $T, T' \in \triangs{n+2}$.
	Then, $T'$ is obtained from $T$ by a bistellar up-flip if and only if
	$E(\Gamma(T)) = E(\Gamma(T')) \cup \{e\}$ for some edge $e\in E(\Gamma(T))$.
	In particular, $T \leq_1 T'$ if and only if $\Gamma(T) \supseteq \Gamma(T')$.
\end{theorem}
\begin{proof}
	Let $T \leq_1 T'$ be related by a bistellar up-flip on the vertices $v_1 < \dots < v_5$,
	that is,
        \begin{align*}
          T &\supseteq \set{
	\{v_1, v_2, v_3, v_4\},
	\{v_1, v_2, v_4, v_5\},
          \{v_2, v_3, v_4, v_5\}} \text{ and}\\
          T' &\supseteq \set{
	\{v_1, v_2, v_3, v_5\},
          \{v_1, v_3, v_4, v_5\}}.
        \end{align*}
	Then, $E(\Gamma(T)) = E(\Gamma(T')) \cup \set{\{v_2, v_4\}}$.
	
	Conversely, let $G,G' \in \persis{n}$ with $E(G) = E(G') \cup \{v, w\}$
	where $v < w$. Then clearly $v+1 < w$.
	Thus, by the \barprop{}, there exists $v < y < w$ with $\set{\{v, y\}, \{y, w\}} \subseteq E(G')$.
	Moreover, $y$ is unique because otherwise the \xprop{} would imply that $\{v, w\} \in E(G')$.
	In fact, the \xprop{} even implies that
	$\ell_G(\{y, w\}) = v$ and $r_G(\{v, y\}) = w$.
	Let further $x := \ell_G(\{v, w\})$ and $z := r_G(\{v, w\})$, thus $\xi_G(\{v, w\}) = \{x, v, w, z\}$.

	We claim that $r_G(e) = r_{G'}(e)$ holds for every edge $e \in E(G') \setminus \{\{v, y\}\}$.
	This is obvious except when $e = \{v, u\}$ for some vertex $u$ with $v < u < w$.
	So assume that is the case.
	Then we must have $u < y$ since $y < u < w$ would force~$G'$ to contain the edge $\{v, w\}$ by the \xprop{}.
	Since $r_G(\{v, u\}) \leq y$, it follows $r_{G'}(\{v, u\}) = r_G(\{v, u\})$, proving the claim.
	Symmetrically one can show $\ell_G(e) = \ell_{G'}(e)$ for every edge $e \in E(G') \setminus \{\{y, w\}\}$.
	
	So we can conclude $\xi_G(e) = \xi_{G'}(e)$ for each edge $e\in E(G')$ with $e\nsubseteq \{v,y,w\}$.
	Thus,
	\[
		\Xi(G') = \left(\Xi(G) \setminus \left\{\xi_G(\{v,y\}), \xi_G(v, w), \xi_G(\{y,w\}) \right\} \right)
		\cup \left\{ \xi_{G'}(\{v, y\}), \xi_{G'}(\{y, w\}) \right\}
	\]
	and it remains to determine $\xi_G(e)$ and $\xi_{G'}(e)$ for $e \subset \{v, y, w\}$.
	
	We claim that $\xi_{G'}(\{v, y\}) = \{x, v, y, z\}$, that is, $\ell_{G'}(\{v, y\})=x$ and $r_{G'}(\{v, y\})=z$.
	First, note that, since $r_G(\{v, y\}) = w$, we must have $r_{G'}(\{v, y\}) > w$ and thus $r_{G'}(\{v, y\}) = r_G(\{v, w\}) = z$.
	Now, if $\ell_{G'}(\{v, y\})> x$, then, by the \xprop{}, $\ell_{G'}(\{v, y\})$ would also be a neighbor of $w$ in $G$, contradicting $x = \ell_G\{v, w\}$.
	To see that $\hat{G'}$ contains the edge $\{x, y\}$, note that otherwise $x$ would have two consecutive neighbors $a, a'$ with $v \leq a < y < a' \leq w$.
	By \cref{thm:consecutive_neighbors}, this implies that~$G'$ contains the edge $\{a, a'\}$ (thus, $\{a, a'\} \neq \{v, w\}$).
	The \xprop{} then implies that~$G'$ also contains the edges $\{v, a'\}$ and $\{a, w\}$.
	If $a \neq v$, then this contradicts $\ell_G(\{y, w\}) = v < a$,
	and if $a' \neq w$, then this contradicts $r_G(\{v, y\}) = w > a'$.
        Thus, we have~$\{x,y\}\in E(\hat{G'})$ implying $\ell_{G'}(\{v, y\})=x$ (and thus also~$\ell_G(\{v, y\})=x$).
	This proves the claim $\xi_{G'}(\{v,y\}) = \{x, v, y, z\}$.
        Moreover, we clearly have $\xi_G(\{v, y\}) = \{x, v, y, w\}$.
	From symmetric arguments it follows that 
	$\xi_{G'}(\{y, w\}) = \{x, y, w, z\}$
	and $\xi_G(\{y, w\}) = \{v, y, w, z\}$.
	        
	Plugging these values into the equality above, we obtain
	\[
		\Xi(G') = (\Xi(G)\setminus \{\{x,v,y,w\},\{x,v,w,z\},\{v,y,w,z\}\}) \cup \{\{x,v,y,z\},\{x,y,w,z\}\},
	\]
	that is, $\Xi(G')$ is obtained from $\Xi(G)$ by a bistellar up-flip on~$x < v < y < w < z$.
\end{proof}

We close with mentioning a connection to higher Bruhat orders.
\textcite[Theorem~3]{Saeedi2015} described a map $\alpha: \persis{n} \to B(n, 2)$, where $B(n,2)$ is the second higher Bruhat order (which is isomorphic to the set of 3-signotopes~\cite{FelsnerSignotopesBruhat}).
Moreover, \textcite{Rambau97} showed an order-preserving map~$f_d \colon B(n,d) \to \text{HST}_1(n+2,d+1)$ from the higher Bruhat order to the first higher Stasheff-Tamari order (see also \cite[Theorem~8.9]{RV12}). It is open whether this map is surjective.
It can be shown that our bijection $\Xi$ equals~$f_2\circ\alpha$, which implies that~$f_2$ is surjective.

\section{Enumerating Triangulations}

The bijection between $\triangs{n+2}$ and $\persis{n}$ has the practical implication that in order to enumerate all triangulations of~$C(n+2, 3)$, one can instead enumerate all persistent graphs on~$n$ vertices.
Since these graphs are combinatorially simpler structures, we can thus improve upon previous enumeration efforts~\cite{JK18}.
We present a simple and efficient algorithm for the enumeration of \pg{}s.

For given $n$, let $\Ee := \binom{[n]}{2} \setminus \set{\{i, i+1\}; i\in[n-1]}$ 
be the set of all potential edges that are not on the obligatory Hamilton path of a \pg{}.
Further, we define $\preceq$~as the colexicographic order on $\Ee$ by setting,
for any $x_1 < y_1$ and $x_2 < y_2$,
\[
\{x_1, y_1\} \preceq \{x_2, y_2\} \iff (y_1 < y_2) \lor (y_1 = y_2 \land x_1 \leq x_2).
\]

Starting from a path~$P_n$, \cref{alg:count_persistent} processes the potential edges $\Ee$ in ascending order
and recurses on each edge, either adding or not adding it to the graph.
Its efficiency arises mainly from the fact that we can quickly identify and skip edges whose addition would violate the X- or \barprop{}.
We remark that, while the listing of \cref{alg:count_persistent} assumes that all inputs are copied upon invocation, it is easy to modify the algorithm such that no copying of $G$ is necessary.

\newcommand{\myfuncname}{PersistentGraphs}
\begin{algorithm}[t]
	\caption{Enumerating persistent graphs on $n$ vertices.}\label{alg:count_persistent}
	\begin{algorithmic}[1]
	\Input{A graph $G=([n],E)$, $k \le n$, and $x < k$,
	such that $E \cap \{e'\in\Ee \mid e' \succ \{x,k\}\} = \emptyset$ and $G$ is persistent except that the edge $\{x, k\}$ (if existing) may violate the \barprop{}.}
	\Output{All persistent supergraphs of $G$ obtainable by adding edges $e' \succeq \{x,k\}$.}
	\Function{\myfuncname{}}{$G$, $k$, $x$}
		\If{$x+1 = k$}
			\If{$k = n$}
				\State output $G$\label{line:k=n}
			\Else
				\State \Call{\myfuncname{}}{$G$, $k+1$, $1$}\label{line:k<n}
			\EndIf
			\State \Return
		\EndIf
		\If{$\{x, k\} \notin E$}\label{line:case1}
			\Comment{edge $\{x, k\}$ does not exist}
			\State $y \gets $ largest neighbor of $x$
			\State \Call{\myfuncname{}}{$G$, $k$, $y$}\label{line:case1recursion}
			\State add $\{x, k\}$ to $E$\label{line:add_e}
		\EndIf
		\For{$y = x+1, \dots, k-1$}\label{line:case2loop}
			\Comment{edge $\{x, k\}$ does exist}
			\If{$\{x, y\} \in E$}
				\State add $\{y, k\}$ to $E$
				\State \Call{\myfuncname{}}{$G$, $k$, $y$}
				\State remove $\{y, k\}$ from $E$ (unless $y = k-1$)
				\State $y \gets$ largest neighbor of $y$
			\EndIf
		\EndFor
	\EndFunction
	\end{algorithmic}
\end{algorithm}

\pagebreak[2]

The following proposition states the correctness.
\begin{proposition}
  Let $G=([n],E)$ be a graph containing a path on~$1,2,\ldots,n$ and let~$e=\{x,k\}$, $1 \le x < k \le n$, be such that the following properties hold:
  \begin{itemize}
    \item $E \cap \{e'\in\Ee \mid e' \succ e\} = \emptyset$.
    \item If~$e\not\in E$, then $G$ is persistent.
    \item If $e\in E$, then either~$G$ is persistent or~$G$ satisfies the \xprop{} and~$e$ is the only edge violating the \barprop{}.
  \end{itemize}
  Then \Call{\myfuncname{}}{$G$, $k$, $x$} outputs exactly all graphs in the set
  \[\mathcal{P}_G^e := \{G'=([n], E')\in \persis{n} \mid (E\subseteq E') \wedge ((E' \setminus E) \subseteq \set{e' \in \Ee; e' \succeq e})\}.\]
\end{proposition}
\begin{proof}
  If $x+1=k$, then clearly~$e\in E$ and $G$ is persistent.
  If now $k=n$, then clearly $\mathcal{P}_G^e=\{G\}$, that is, \cref{line:k=n} is correct.
  If $k < n$, then~$\mathcal{P}_G^e = \mathcal{P}_G^{\{1,k+1\}}$. Thus, \cref{line:k<n} is correct.

  Now assume that $x + 1 < k$.
  If $e \notin E$ (\cref{line:case1}), then~$G$ is persistent.
  The set~$\mathcal{P}_G^e$ can be partitioned into two sets:
  \begin{align*}
    A&:=\{G'=([n],E')\in \mathcal{P}_G^e\mid e \not\in E'\} \text{ and}\\
    B&:=\{G'=([n],E')\in \mathcal{P}_G^e\mid e \in E'\}.
  \end{align*}
  Consider a graph~$G'=([n],E')\in A$. Let~$y$ be the largest neighbor of~$x$ in~$G$ and note that~$y < k$ since~$E$ does not contain any edge~$e'$ with~$e' \succ e$.
  Due to the \xprop{}, $E'$ does not contain any edge $\{x', k\}$ with $x < x' < y$.
	Thus, $A = \mathcal{P}_G^{\{y,k\}}$ and all these graphs are produced by the recursive call in \cref{line:case1recursion}.
        As regards the set~$B$, note that~$B = \mathcal{P}_{G+e}^e$, where~$G+e := ([n],E\cup\{e\})$.
        Thus, we add $e$ to $E$ in \cref{line:add_e} and then handle this case in \cref{line:case2loop}.
  
	If $e\in E$, then, for every~$G'=([n],E')\in \mathcal{P}_G^e$, there must be a minimal vertex $y$ with $x < y < k$ and $\set{\{x, y\}, \{y, k\}} \subseteq E'$ (by the \barprop{}).
	Since $\{x, y\} \prec e$, it follows that $\{x,y\}\in E$.
        Hence, $y$ has to be a neighbor of~$x$ in~$G$ with~$x < y < k$.
	Furthermore, any neighbor~$y'$ of~$x$ with $x < y' < y$ cannot have a neighbor larger than~$y$, because the \xprop{} would otherwise imply that also $\{y', k\} \in E'$, contradicting the minimality of~$y$.
        That is, $y$~can only be neighbor of~$x$ such that no other neighbor~$y'$ of~$x$ with~$x < y' < y$ has a neighbor larger than~$y$. Let~$Y$ denote the vertex set containing all these possible candidates.
        The for-loop in \cref{line:case2loop} iterates exactly over the candidates in~$Y$.
	For a given $y\in Y$, let~$A_y\subseteq\mathcal{P}_G^e$ be the subset of graphs~$G'=([n],E')$, where~$y$ is the minimal vertex with~$x<y<k$ and $\set{\{x, y\}, \{y, k\}} \subseteq E'$, and note that $A_y = \mathcal{P}_{G+\{y,k\}}^{\{y,k\}}$. Moreover,~$\{A_y\mid y\in Y\}$ is clearly a partition of~$\mathcal{P}_G^e$.
        Hence, calling \Call{\myfuncname{}}{$G$, $k$, $y$} for each possible $y$, outputs exactly the graphs in $\mathcal{P}_G^e$.
\end{proof}

\begin{corollary}\label{thm:alg_cor}
	Let the graph~$G=([n], E)$ be the path on vertices~$1,2,\ldots,n$.
	Then \Call{\myfuncname{}}{$G$, $2$, $1$} outputs exactly $\mathcal{P}_G^{\{1,2\}}=\persis{n}$.
\end{corollary}

By \cref{thm:alg_cor}, we can use \cref{alg:count_persistent} to efficiently count the number of elements of $\persis{n}$ and thus of $\triangs{n+2}$.
The results for~$n \le 16$ are listed in \cref{tab:counts}.
The computations\footnote{Implementation available at \url{https://www.akt.tu-berlin.de/menue/software}.} were performed using an Intel~Xeon~W-2125~CPU.

\begin{table}
\centering
\caption{Number of \pg{}s with~$n\le 16$ vertices. The values for $n\leq13$ were already known~\cite{JK18}.}\label{tab:counts}
\begin{tabular}{crr}
	$n$ & $\abs{\persis{n}} = \abs{\triangs{n+2}}$ & computation time\\
	\hline
	1 & 1 & $<$\,0.1\,s \\
	2 & 1 & $<$\,0.1\,s \\
	3 & 2 & $<$\,0.1\,s \\
	4 & 6 & $<$\,0.1\,s \\
	5 & 25 & $<$\,0.1\,s \\
	6 & 138 & $<$\,0.1\,s \\
	7 & 972 & $<$\,0.1\,s \\
	8 & 8\,477 & $<$\,0.1\,s \\
	9 & 89\,405 & $<$\,0.1\,s \\
	10 & 1\,119\,280 & $<$\,0.1\,s \\
	11 & 16\,384\,508 & 0.15\,s \\
	12 &    276\,961\,252    & 2\,s \\
        13 & 5\,349\,351\,298    & 30\,s \\
        \hline
	14 &  116\,985\,744\,912  &  12\,m \\
	15 & 2\,873\,993\,336\,097 & 4\,h~30\,m \\
	16 & 78\,768\,494\,976\,617    & 4\,d~23\,h~2\,m \\
\end{tabular}
\end{table}

\section{Conclusion}
Our results yield further insights into the structure of the triangulations of the 3\nobreakdash-dimen\-sional cyclic polytope by relating their 1-skeleton to persistent graphs.
It remains open to characterize the structure of the $\lfloor d/2\rfloor$-skeleton for arbitrary odd dimension~$d$.
It is also open whether a closed formula for the number of triangulations of~$C(n+2,3)$ can be given~\cite[Open Problem~9.2]{RV12}.

\paragraph{Acknowledgment.}
We thank Lito Goldmann for his work on enumerating persistent graphs, which led us to the discovery of the bijection.
We further thank Andr\'{e} Nichterlein for helpful initial discussions.
We also thank the anonymous reviewers of Combinatorica for their valuable comments.

{
\raggedright
\printbibliography
}

\end{document}